\documentclass{article}
\usepackage{amssymb}
\usepackage[english]{babel}
\usepackage{braket}
\usepackage[letterpaper,top=2cm,bottom=2cm,left=3cm,right=3cm,marginparwidth=1.75cm]{geometry}
\usepackage{amsmath}
\usepackage{graphicx}
\usepackage[colorlinks=true, allcolors=blue]{hyperref}

\newtheorem{theorem}{Theorem}[section]
\newtheorem{definition}[theorem]{Definition}
\newtheorem{lemma}[theorem]{Lemma}
\newtheorem{proof}[theorem]{Proof}
\newtheorem{corollary}[theorem]{Corollary}
\newtheorem{conjecture}[theorem]{Conjecture}
\title{Quantum Algorithms for Gowers Norm Estimation, Polynomial Testing, and Arithmetic Progression Counting over Finite Abelian Groups}

\author{
  En-Jui Kuo\thanks{Department of Electrophysics, National Yang Ming Chiao Tung University, Hsinchu, Taiwan, R.O.C.}
}

\date{}

\begin{document}

\maketitle

\begin{abstract}
We propose a family of quantum algorithms for estimating Gowers uniformity norms \( U^k \) over finite abelian groups and demonstrate their applications to testing polynomial structure and counting arithmetic progressions. Building on recent work for estimating the \( U^2 \)-norm over \( \mathbb{F}_2^n \), we generalize the construction to arbitrary finite fields and abelian groups for higher values of \( k \). Our algorithms prepare quantum states encoding finite differences and apply Fourier sampling to estimate uniformity norms, enabling efficient detection of structural correlations.

As a key application, we show that for certain degrees \( d = 4, 5, 6 \) and under appropriate conditions on the underlying field, there exist quasipolynomial-time quantum algorithms that distinguish whether a bounded function \( f(x) \) is a degree-\( d \) phase polynomial or far from any such structure. These algorithms leverage recent inverse theorems for Gowers norms, together with amplitude estimation, to reveal higher-order algebraic correlations.

We also develop a quantum method for estimating the number of 3-term arithmetic progressions in Boolean functions \( f : \mathbb{F}_p^n \to \{0,1\} \), based on estimating the \( U^2 \)-norm. Though not as query-efficient as Grover-based counting, our approach provides a structure-sensitive alternative aligned with additive combinatorics.

Finally, we demonstrate that our techniques remain valid under certain quantum noise models, due to the shift-invariance of Gowers norms. This enables noise-resilient implementations within the NISQ regime and suggests that Gowers-norm-based quantum algorithms may serve as robust primitives for quantum property testing, learning, and pseudorandomness.
\end{abstract}

\section{Introduction}

Fourier analysis has long played a central role across mathematics, physics, computer science, and engineering. Its discrete counterpart—Fourier analysis over finite groups—has been developed through group representation theory~\cite{Serre1977, FultonHarris1991}, enabling applications in molecular dynamics~\cite{muller2024symmetry}, crystallography~\cite{Cotton1990}, and quantum field theory~\cite{Georgi1999}. In quantum computing, this theory underpins the quantum Fourier transform (QFT), a fundamental subroutine in many celebrated quantum algorithms.

Notably, the QFT enables the exponential speedup in Shor's factoring algorithm~\cite{Shor1997}, and plays a central role in the Hidden Subgroup Problem (HSP) framework~\cite{Jozsa2001}, which abstracts the core structure behind several quantum speedups, including for period finding, discrete log, and dihedral hidden shift problems~\cite{Childs2007}. These advances illustrate how symmetry and spectral structure—captured via Fourier transforms—can be exploited in quantum algorithms.

In parallel, classical theoretical computer science has leveraged Fourier analysis over $\mathbb{F}_2^n$ to great effect. The BLR linearity test~\cite{Blum1993}, for example, introduced the modern theory of property testing, and inspired subsequent tools in coding theory~\cite{Arora1998}, learning theory~\cite{KushilevitzMansour1993}, and pseudorandomness~\cite{ODonnellBook}. These results typically hinge on the behavior of low-degree Fourier coefficients, or linear characters.

However, such linear Fourier techniques have intrinsic limitations. For instance, Roth's theorem~\cite{Roth1953} uses linear Fourier analysis to prove that dense subsets of the integers contain 3-term arithmetic progressions (3-APs), but Szemerédi’s theorem for 4-APs lies beyond the reach of linear methods. This prompted the development of \emph{higher-order Fourier analysis}, initiated by Gowers~\cite{Gowers2001} and later extended by Green, Tao, Ziegler, Hatami, Lovett, and others~\cite{GreenTaoU3, GTZ2012, hatami2019higher}. At the heart of this theory are the Gowers uniformity norms \( \|f\|_{U^k} \), which detect hidden polynomial structure of degree \( < k \). The inverse theorem asserts that if \( \|f\|_{U^k} \) is large, then \( f \) must correlate with a degree-\( (k-1) \) phase polynomial.

Despite their foundational impact in additive combinatorics and number theory, Gowers norms have received relatively little attention in the context of quantum computing. 

Recently, \cite{jothishwaran2020quantum, bera2021quantum} introduced a quantum algorithm to estimate the Gowers \( U^2 \)-norm over \( \mathbb{F}_2^n \), laying the groundwork for a norm-based approach to quantum property testing. Their work focused primarily on Boolean functions and testing quadraticity, but did not extend to higher-order uniformity or more general finite abelian groups. Gowers norms have also been employed in other quantum contexts, such as self-testing protocols~\cite{westdorpquantum}, and bounding the complexity of quantum states with low stabilizer rank~\cite{mehraban2025improved}, illustrating the growing relevance of higher-order Fourier analysis in quantum information theory.

In this paper, we build on this direction and develop quantum algorithms to estimate \( U^k \)-norms over arbitrary finite abelian groups \( G = \mathbb{F}_p^n \). Unlike in the Boolean setting, functions over such groups can exhibit genuine higher-order structure, making this regime nontrivial and rich. Our work reveals that the tools of higher-order Fourier analysis can be natively implemented via quantum circuits, and used to test algebraic structure beyond what is classically feasible.

\paragraph{Our contributions.} We summarize our main results as follows:
\begin{itemize}
    \item We construct quantum algorithms for estimating the Gowers uniformity norm \( \|f\|_{U^k} \) for any \( k \geq 2 \), using phase oracles and quantum Fourier transforms over \( \mathbb{F}_p^n \). Our circuits generalize the \( U^2 \)-norm algorithm of~\cite{jothishwaran2020quantum} to arbitrary prime fields. The main result appears in Theorem~\ref{thm:linear}.
    
    \item Leveraging recent inverse theorems for Gowers norms~\cite{milicevic2022quantitative, milicevic2024quasipolynomial}, we show that if \( f \) is \( \varepsilon \)-far from all degree-\( d \) phase polynomials, then \( \|f\|_{U^{d+1}} \) must be small. This enables quasipolynomial-time quantum algorithms for approximate polynomial testing, extending prior results that were restricted to \( d = 1, 2 \)~\cite{Gharibian2020, mehraban2025improved}. The main result is presented in Theorem~\ref{thm:qn}.
    
    \item We apply our norm estimation algorithm to count 3-term arithmetic progressions (3-APs) in dense subsets. Specifically, \( \|f\|_{U^2} \) provides an estimate of the 3-AP count with bounded variance. This yields a structured alternative to Grover-based quantum counting, with interpretability grounded in additive combinatorics. See Subsection~\ref{subsec:c} for details.
\end{itemize}

Finally, we show that our algorithms remain robust in the NISQ regime. Due to the shift-invariance of Gowers norms, our Fourier-sampling-based methods tolerate unknown, adversarial translation errors, and naturally operate under noisy oracle models such as \( \tfrac{1}{2}\text{BQP} \)~\cite{jacobs2024space}. This suggests a new family of quantum property testers and learning primitives built on Gowers norms, with practical potential for early quantum devices.

\paragraph{Outline.}  
In Section~\ref{sec:prelim}, we review Fourier analysis over finite abelian groups and define the Gowers norms. Section~\ref{sec:Q} introduces the quantum model and necessary primitives. Section~\ref{sec:algorithm} presents our core algorithms for norm estimation and structure detection. Section~\ref{sec:applications} covers applications to polynomial testing and arithmetic progression counting. In Section~\ref{sec:nisq}, we analyze noise-resilience under \( \frac{1}{2}\text{BQP} \) and \( \text{BQP}_\lambda^0 \). We conclude in Section~\ref{sec:conclusion} with a discussion of open problems and future directions.

\section{Gowers norms over Abelian group}\label{sec:prelim}
\subsection{Functions on Finite Abelian Groups and Fourier Analysis}

Let \( G \) be a finite abelian group written additively, with identity element \( 0 \). The set of complex-valued functions on \( G \) is denoted by \( \mathcal{F}(G) := \{ f: G \to \mathbb{C} \} \). The group of characters of \( G \), denoted \( \widehat{G} \), consists of all group homomorphisms \( \gamma: G \to \mathbb{C}^\times \) such that \( \gamma(x + y) = \gamma(x)\gamma(y) \) for all \( x, y \in G \). When \( G \) is finite abelian, the dual group \( \widehat{G} \) is isomorphic to \( G \) itself, i.e., \( \widehat{G} \cong G \), which allows us to represent the functions on \( G \) as linear combinations of the group characters.

The inner product on \( \mathcal{F}(G) \) is defined by:
\begin{equation}
\langle f, g \rangle := \mathbb{E}_{x \in G} f(x)\overline{g(x)} = \frac{1}{|G|} \sum_{x \in G} f(x)\overline{g(x)},
\end{equation}
which measures the similarity between two functions \( f \) and \( g \) over the group \( G \), taking into account the structure of \( G \) through the summation over its elements. The convention we use is the following, where the expectation denotes averaging over the finite abelian group $G$:
\begin{align}
    \mathbb{E}_{x \in G} &:= \frac{1}{|G|} \sum_{x \in G}, \\
    \mathbb{E}_{x,a \in G} &:= \frac{1}{|G|^2} \sum_{x, a \in G},
\end{align}
and similarly for more indices. This is called group averaging.

\subsubsection*{Fourier Expansion and Parseval's Identity}
The Fourier coefficient of \( f \in \mathcal{F}(G) \) at character \( \gamma \in \widehat{G} \) is defined as:
\begin{equation}
\widehat{f}(\gamma) := \mathbb{E}_{x \in G} f(x)\overline{\gamma(x)},
\end{equation}
which quantifies how much of the function \( f \) "matches" with the character \( \gamma \). This Fourier coefficient plays a central role in expressing \( f \) in terms of the group characters.

Thus, \( f \) admits the Fourier expansion:
\begin{equation}
f(x) = \sum_{\gamma \in \widehat{G}} \widehat{f}(\gamma) \gamma(x),
\end{equation}
which expresses the function \( f \) as a sum of its projections onto the characters of the group.

The Parseval identity holds:
\begin{equation}
\sum_{\gamma \in \widehat{G}} |\widehat{f}(\gamma)|^2 = \mathbb{E}_{x \in G} |f(x)|^2=\frac{1}{|G|}\sum_{x \in G} |f(x)|^2,
\end{equation}
which asserts that the total energy (or norm) of the function \( f \) over the group is equal to the sum of the energies of its Fourier components. This identity is a key result in Fourier analysis, indicating that no information is lost when decomposing the function into its Fourier components.

\subsubsection*{Convolution and Autocorrelation}
For \( f, g \in \mathcal{F}(G) \), the convolution product is defined as:
\begin{equation}
(f \ast g)(x) := \mathbb{E}_{y \in G} f(y)g(x - y),
\end{equation}
which is a natural way of combining the two functions \( f \) and \( g \) over the group. This operation measures how one function "overlaps" with the other as one is shifted over the group. The Fourier transform of the convolution satisfies:
\begin{equation}
\widehat{f \ast g}(\gamma) = \widehat{f}(\gamma) \cdot \widehat{g}(\gamma),
\end{equation}
which reflects that the convolution in the time domain corresponds to pointwise multiplication in the Fourier domain.

The autocorrelation of \( f \) at \( a \in G \) is defined as:
\begin{equation}
\mathrm{Corr}_f(a) := \mathbb{E}_{x \in G} f(x) \overline{f(x + a)},
\end{equation}
which measures the correlation of the function \( f \) with its shifted version \( f(x + a) \). Autocorrelation is useful for understanding the structure of a function by examining how much it "resembles" itself when shifted by a certain amount.

\subsection{Gowers Norm}
We first introduce the finite difference
\subsubsection*{Finite Differences}
The multiplicative finite difference of \( f \in \mathcal{F}(G) \) at \( a \in G \) is defined as:
\begin{equation}
\Delta_a f(x) := f(x)\overline{f(x + a)},
\end{equation}
which measures how much \( f(x) \) changes when shifted by \( a \) in the group \( G \). This operation helps quantify how the function \( f \) behaves under small shifts, and is central to defining the Gowers uniformity norms.

The \( k \)-fold iterated difference operator is defined as:
\begin{equation}
\Delta_{h_1, \dots, h_k} f(x) := \prod_{S \subseteq [k]} \mathcal{C}^{|S|} f\left(x + \sum_{i \in S} h_i\right),
\end{equation}
where \( \mathcal{C}^{|S|} \) represents the application of complex conjugation if \( |S| \) (the size of the subset \( S \)) is odd. This operator iteratively applies the finite difference across multiple dimensions, capturing higher-order correlations of the function over shifts by \( h_1, \dots, h_k \). 

This formulation generalizes the standard Gowers derivative in \( \mathbb{F}_2^n \), and leads to the definition of the Gowers uniformity norms \( \|f\|_{U^k} \) for arbitrary finite abelian groups. These norms are crucial for detecting structured patterns in functions defined over such groups, and they form the basis for many results in additive combinatorics and theoretical computer science.

Let \( G \) be a finite abelian group and let \( f: G \to \mathbb{C} \) be a bounded function. The Gowers \( U^k \) norm of \( f \) is defined by:
\begin{equation}
    \|f\|_{U^k}^{2^k} := \mathbb{E}_{x, h_1, \dots, h_k \in G} \prod_{\omega \in \{0,1\}^k} \mathcal{C}^{|\omega|} f\left(x + \omega \cdot h\right),
\end{equation}
where:
\begin{itemize}
    \item \( \omega = (\omega_1, \dots, \omega_k) \in \{0,1\}^k \) is a binary vector that determines the combination of shifts,
    \item \( \omega \cdot h := \sum_{j=1}^k \omega_j h_j \in G \) represents the weighted sum of shifts based on the components of \( \omega \),
    \item \( |\omega| := \sum_{j=1}^k \omega_j \) counts the number of nonzero components of \( \omega \),
    \item \( \mathcal{C}^{|\omega|} \) applies complex conjugation \( |\omega| \) times, i.e., conjugates if \( |\omega| \) is odd.
\end{itemize}
This norm captures the behavior of \( f \) under all combinations of \( k \)-dimensional shifts and complex conjugation, and is central to detecting whether a function has high uniformity with respect to the group structure.

\subsection*{Example: The \( U^2 \) Norm}
In the case \( k = 2 \), the \( U^2 \) norm is given by:
\begin{equation}\label{eq:U2}
    \|f\|_{U^2}^4 = \mathbb{E}_{x, h_1, h_2 \in G} \left[
        f(x)
        \overline{f(x + h_1)}
        \overline{f(x + h_2)}
        f(x + h_1 + h_2)
    \right]= \frac{1}{|G|^3}\sum_{x,h_1,h_2 \in G}f(x)
        \overline{f(x + h_1)}
        \overline{f(x + h_2)}
        f(x + h_1 + h_2).
\end{equation}
This expression measures the correlation of the function \( f \) at different points in \( G \) by evaluating it at several shifted positions. The goal is to determine whether these values exhibit any structure that suggests \( f \) is close to a polynomial of degree 1.

Let \( \widehat{G} \) denote the dual group of \( G \), and define the Fourier transform of \( f \) at \( \gamma \in \widehat{G} \) as:
\begin{align}
\hat{f}(\gamma) := \mathbb{E}_{x \in G} f(x) \overline{\gamma(x)}.
\end{align}
Then, the \( U^2 \) norm can be written as:
\begin{equation}
    \|f\|_{U^2}^4 = \sum_{\gamma \in \widehat{G}} |\hat{f}(\gamma)|^4,
\end{equation}
which expresses the norm as the sum of the Fourier coefficients of \( f \) raised to the fourth power. This representation is key for analyzing the structure of the function \( f \) in the Fourier domain.

\begin{proof}
We now provide a standard proof of the \( U^2 \) norm expression by expanding \( f \) in terms of its Fourier coefficients. Starting from the definition:
\begin{equation}
\|f\|_{U^2}^4 = \mathbb{E}_{x,a,b \in G} f(x)\overline{f(x+a)}\overline{f(x+b)}f(x+a+b),
\end{equation}
we expand each instance of \( f(x) \) using its Fourier representation:
\begin{align*}
f(x) &= \sum_{\alpha} \widehat{f}(\alpha) \alpha(x), \\
f(x+a) &= \sum_{\beta} \widehat{f}(\beta) \beta(x+a) = \sum_{\beta} \widehat{f}(\beta) \beta(x)\beta(a), \\
f(x+b) &= \sum_{\gamma} \widehat{f}(\gamma) \gamma(x+b) = \sum_{\gamma} \widehat{f}(\gamma) \gamma(x)\gamma(b), \\
f(x+a+b) &= \sum_{\delta} \widehat{f}(\delta) \delta(x+a+b) = \sum_{\delta} \widehat{f}(\delta) \delta(x)\delta(a)\delta(b).
\end{align*}
After substituting these expansions, the full expression becomes:
\begin{align}
\|f\|_{U^2}^4 &= \mathbb{E}_{x,a,b} 
\sum_{\alpha} \widehat{f}(\alpha) \alpha(x) 
\sum_{\beta} \overline{\widehat{f}(\beta)} \overline{\beta(x)} \overline{\beta(a)} \\
&\qquad \times \sum_{\gamma} \overline{\widehat{f}(\gamma)} \overline{\gamma(x)} \overline{\gamma(b)}
\sum_{\delta} \widehat{f}(\delta) \delta(x) \delta(a) \delta(b).
\end{align}
Simplifying the sums and using orthogonality of the characters, we obtain the final expression:
\begin{equation}
\|f\|_{U^2}^4 = \sum_{\delta \in \widehat{G}} |\widehat{f}(\delta)|^4.
\end{equation}
\end{proof}

However, the \( U^3 \) norm does not have as simple a relationship as the \( U^2 \) norm due to its involvement with higher-order correlations. For example, the \( U^3 \) norm of a function \( f \) on a finite abelian group \( G \) is defined as:
\begin{align}
\|f\|_{U^3}^8 := \mathbb{E}_{x,h_1,h_2,h_3 \in G} \prod_{\omega \in \{0,1\}^3} \mathcal{C}^{|\omega|} f\left(x + \omega_1 h_1 + \omega_2 h_2 + \omega_3 h_3\right),
\end{align}
which is a more complex form compared to the \( U^2 \) norm. It can be expressed in terms of Fourier coefficients as:
\begin{align}
\|f\|_{U^3}^8 = \sum_{\substack{
\gamma_1 + \gamma_2 = \gamma_3 + \gamma_4 \\
\gamma_1 + \gamma_5 = \gamma_3 + \gamma_7 \\
\gamma_5 + \gamma_6 = \gamma_7 + \gamma_8
}} \widehat{f}(\gamma_1) \overline{\widehat{f}(\gamma_2)} \overline{\widehat{f}(\gamma_5)} \widehat{f}(\gamma_6) \overline{\widehat{f}(\gamma_3)} \widehat{f}(\gamma_4) \widehat{f}(\gamma_7) \overline{\widehat{f}(\gamma_8)}.
\end{align}
This expression involves more intricate interactions between the Fourier coefficients and higher-order combinations of shifts. While \( U^2 \) norm deals with pairwise correlations, \( U^3 \) norm generalizes this to capture interactions involving three variables, making it more challenging to compute and analyze.

One of the key motivations for studying the Gowers norm is its ability to detect polynomial structure. Specifically, functions of the form \( f(x) = \omega^{P(x)} \), where \( P \) is a low-degree polynomial, have maximal Gowers norm.
\begin{theorem}[Direct Theorem: Phase Polynomials Maximize Gowers Norm]
Let \( P: \mathbb{F}_p^n \to \mathbb{C} \) be a polynomial of degree at most \( d-1 \), and let \( f(x) := \omega^{P(x)} \), where \( \omega = e^{2\pi i/p} \). Then
\begin{align}
\|f\|_{U^d} = 1.
\end{align}
\end{theorem}

\begin{proof}[Sketch of proof, see \cite{tao2012higher, hatami2019higher}]
The Gowers \( U^d \) norm involves a \( 2^d \)-fold multiplicative finite difference operator. For a degree-\( (d-1) \) polynomial \( P \), this operator vanishes identically, i.e.,
\begin{align}
\Delta_{h_1,\dots,h_d} \omega^{P(x)} = 1
\end{align}
for all \( x, h_1, \dots, h_d \), since all higher-order discrete derivatives cancel out. As a result, the expectation in the definition of \( \|f\|_{U^d} \) equals 1.
\end{proof}

This property makes the Gowers norm a useful tool for detecting hidden polynomial structure in bounded functions. If a function has large Gowers norm, one may ask whether it correlates with a phase polynomial—a question addressed by the inverse theorem in the next section.

\subsection{Inverse Theorem}
We begin by discussing the Original Inverse Gowers Conjecture, which is central to understanding the connection between Gowers norms and polynomial structure:

\begin{conjecture}[Original Inverse Gowers Conjecture]
Let \( p \geq 2 \) be a fixed prime and \( d \geq 1 \). For any bounded function \( F : \mathbb{F}_p^n \to \mathbb{C} \) with \( \|F\|_{\infty} \leq 1 \), if the Gowers norm \( \|F\|_{U^{d+1}} \geq \varepsilon \), then there exists a degree-\( d \) polynomial \( P : \mathbb{F}_p^n \to \mathbb{F}_p \) such that
\begin{align}
|\langle F, e^{2\pi i P/p} \rangle| \geq \delta(p, d, \varepsilon).
\end{align}
\end{conjecture}

For \( d = 1 \), the inverse theorem has been proved (see \cite{tao2012higher}), and it states that if \( \|F\|_{U^2} \geq \varepsilon \), then there exists a linear polynomial \( P \) such that:
\begin{align}
|\langle F, e^{2\pi i P/p} \rangle| \geq \varepsilon^2.
\end{align}

For \( d = 2 \), a similar result holds, as proven in \cite{hatami2019higher}, which states that if \( \|F\|_{U^3} \geq \varepsilon \), then there exists a degree-2 polynomial \( P \) such that:
\begin{align}
|\langle F, e^{2\pi i P/p} \rangle| \geq \delta.
\end{align}
Here, \( \delta = \delta(p, \varepsilon) > 0 \) is a constant depending on \( p \) and \( \varepsilon \).

However, when \( d = 3 \), the original Gowers conjecture faces complications. Researchers have found that for \( d = 3 \), there exist generalized polynomials that satisfy the finite difference structure but are not standard polynomials. These generalized polynomials do not fit into the classical framework that the conjecture is based on, thus leading to a failure of the conjecture for \( d = 3 \).

This issue arises because generalized polynomials capture higher-order correlations that cannot be described by simple polynomials of degree \( d-1 \). Therefore, the Gowers conjecture does not hold in its original form for \( d = 3 \) and beyond. To handle this issue, we just need to put such generalized polynomial into the set or applies when the characteristic of the field is greater than \( d \), and provides a clean solution that circumvents the need for generalized polynomials.

Finally we have the following theorem (Theorem 1.5.3 in \cite{tao2012higher}):
\begin{theorem}[1\% Inverse Theorem for \( U_{d+1} \) Norm (Theorem 1.5.3 in \cite{tao2012higher})] \label{thm:inverse}
Let \( \mathbb{F} \) be a field of characteristic greater than \( d \), and let \( V = \mathbb{F}_p^n \). Suppose \( f : V \to \mathbb{C} \) satisfies \( \|f\|_\infty \leq 1 \) and \( \|f\|_{U_{d+1}} \geq \varepsilon \). Then there exists a degree-\( d \) polynomial \( P : V \to \mathbb{F}_p \) such that
\begin{align}
|\langle f, \omega^{P} \rangle| \geq \delta,
\end{align}
for some constant \( \delta = \delta(p, d, \varepsilon) > 0 \).
\end{theorem}

This theorem is crucial because it connects the Gowers \( U^k \) norm with polynomials of degree \( k-1 \) in domains where the characteristic of the field is greater than \( d \). It simplifies the analysis by avoiding the complications of generalized polynomials and allows us to work directly with standard polynomials.

By using Theorem \ref{thm:inverse}, we can avoid the problems associated with generalized polynomials. If we assume that the characteristic of the field is greater than \( d \), we can apply this result directly and detect polynomial structure in functions with large Gowers norms without worrying about the complications that arise for \( d \geq 3 \). Thus, for fields with characteristic greater than \( d \), we can rely on Theorem \ref{thm:inverse} to handle the inverse problem effectively, which also has deep implications for both classical and quantum algorithms.



\section{Quantum Computation}\label{sec:Q}

We now describe the key quantum operations used in our algorithm. These primitives are standard in the study of quantum algorithms over finite abelian groups. For a more detailed exposition, we refer the reader to Nielsen and Chuang~\cite{NC10}, Høyer~\cite{Hoyer1997}.

Let \( G \) be a finite abelian group of order \( N \). We define the Hilbert space
\begin{align}
\mathcal{H}_G := \text{span}_\mathbb{C} \{ \ket{g} : g \in G \},
\end{align}
with orthonormal basis states indexed by elements \( g \in G \). A pure quantum state over \( G \) is a unit vector in \( \mathcal{H}_G \) and can be written as
\begin{align}
\ket{\psi} = \sum_{g \in G} \alpha_g \ket{g}, \quad \text{with } \sum_{g \in G} |\alpha_g|^2 = 1.
\end{align}

We now define the three key operations that act on these states.

\begin{definition}[Phase Oracle \( U_f \)]
Let \( f : G \to \mathbb{C} \) be a function with \( |f(x)| = 1 \) for all \( x \in G \). The phase oracle is the unitary operator defined by:
\begin{align}
U_f \ket{x} = f(x) \ket{x}.
\end{align}
This oracle is commonly assumed in the quantum query model. When \( f(x) = \omega^{P(x)} \) for a polynomial \( P : G \to \mathbb{F}_p \), the oracle can be implemented efficiently using modular arithmetic over \( \mathbb{F}_p \).
\end{definition}

\begin{definition}[Controlled Group Addition \( \mathrm{CADD}_{x \to y} \)]
Let \( G \) be a finite abelian group written additively. The controlled group addition gate is defined by:
\begin{align}
\mathrm{CADD}_{x \to y} \ket{x} \ket{y} = \ket{x} \ket{x + y}.
\end{align}
This gate performs a reversible addition from the control register \( x \) to the target register \( y \). Its inverse is given by:
\begin{align}
\mathrm{CADD}^{-1}_{x \to y} \ket{x} \ket{y} = \ket{x} \ket{y - x}.
\end{align}
When \( G = \mathbb{F}_2^n \), this gate reduces to a multi-qubit CNOT gate.
\end{definition}

\begin{definition}[Quantum Fourier Transform \( \mathrm{QFT}_G \)]
Let \( G \) be a finite abelian group and \( \widehat{G} \) its dual group. The quantum Fourier transform over \( G \) is the unitary transformation defined by:
\begin{align}
\mathrm{QFT}_G \ket{x} := \frac{1}{\sqrt{|G|}} \sum_{\gamma \in \widehat{G}} \gamma(x) \ket{\gamma},
\end{align}
where \( \gamma(x) \) is the character evaluation of \( x \in G \). The inverse Fourier transform is given by:
\begin{align}
\mathrm{QFT}_G^{-1} \ket{\gamma} = \frac{1}{\sqrt{|G|}} \sum_{x \in G} \overline{\gamma(x)} \ket{x}.
\end{align}
\end{definition}

\paragraph{Special Case: \( G = \mathbb{F}_2^n \)}

In this case, \( \widehat{G} \cong G \), and characters are given by \( \gamma(x) = (-1)^{\langle \gamma, x \rangle} \). The Fourier transform becomes the Hadamard transform:
\begin{align}
\mathrm{QFT}_{\mathbb{F}_2^n} \ket{x} = \frac{1}{\sqrt{2^n}} \sum_{\gamma \in \mathbb{F}_2^n} (-1)^{\langle \gamma, x \rangle} \ket{\gamma},
\end{align}
which equals \( H^{\otimes n} \), the tensor product of \( n \) Hadamard gates.

\paragraph{General Case: \( G = \mathbb{F}_p^n \)}

For general \( p \), characters are defined as \( \gamma(x) = \omega^{\langle \gamma, x \rangle} \), where \( \omega = e^{2\pi i/p} \) is a primitive \( p \)-th root of unity. The Fourier transform becomes:
\begin{align}
\mathrm{QFT}_{\mathbb{F}_p^n} \ket{x} = \frac{1}{\sqrt{p^n}} \sum_{\gamma \in \mathbb{F}_p^n} \omega^{\langle \gamma, x \rangle} \ket{\gamma},
\end{align}
and can be implemented efficiently using known methods from quantum modular arithmetic~\cite{NC10,Hoyer1997}.

\subsection{Summary of Quantum Primitives Used}

Throughout this paper, we assume the availability of the following quantum operations:

\begin{itemize}
    \item Phase oracle \( U_f \): \( \ket{x} \mapsto f(x) \ket{x} \)
    \item Controlled group addition gate \( \mathrm{CADD}_{x \to y} \): \( \ket{x} \ket{y} \mapsto \ket{x} \ket{x + y} \)
    \item Quantum Fourier transform \( \mathrm{QFT}_G \): \( \ket{x} \mapsto \frac{1}{\sqrt{|G|}} \sum_{\gamma \in \widehat{G}} \gamma(x) \ket{\gamma} \)
\end{itemize}

These primitives provide the foundation for Gowers-norm estimation, testing polynomial structure, and implementing quantum learning algorithms over \( \mathbb{F}_p^n \).

\section{Quantum Estimation of Gowers Norms}\label{sec:algorithm}

We now describe quantum algorithms for estimating the Gowers \( U^d \) norm, a key quantity in higher-order Fourier analysis. These algorithms rely on phase oracles and group structure over finite abelian groups.

\subsection{Quantum Algorithm to Estimate Gowers \( U^2 \) Norm over Finite Abelian Groups}

Let \( G \) be a finite abelian group of order \( N = |G| \), and let \( f: G \to \mathbb{C} \) and $|f|=1$. The Gowers \( U^2 \) norm of \( f \) is defined as
\begin{equation}
\label{eq:u2-def}
\|f\|_{U^2}^4 := \mathbb{E}_{x,a,b \in G} \left[ f(x) \overline{f(x+a)} \overline{f(x+b)} f(x+a+b) \right] = \mathbb{E}_{x,a,b} \left[ \Delta_{a,b}f(x) \right],
\end{equation}
where \( \Delta_{a,b}f(x) := f(x) \overline{f(x+a)} \overline{f(x+b)} f(x+a+b) \).

\begin{theorem}
Let \( G \) be a finite abelian group and \( f: G \to \mathbb{C} \). There exists a quantum algorithm that uses \( 4 \) queries to the phase oracle \( U_f \) (or its adjoint \( U_f^\dagger \)), and performs \( 3 \) quantum Fourier transforms over \( G \), such that the probability of measuring the all-zero state \( \ket{0}^{\otimes 3} \) equals \( \|f\|_{U^2}^8 \). The overall gate complexity is
\begin{align}
\mathcal{O}\left( 4 C_f + 3 \cdot \log^2 |G| \right),
\end{align}
where \( C_f \) denotes the cost of implementing the oracle \( U_f \).
\end{theorem}

\begin{proof}
To estimate \( \|f\|_{U^2} \), we begin with the uniform superposition:
\begin{align}
\ket{\psi_0} = \frac{1}{N^{3/2}} \sum_{x,a,b \in G} \ket{x} \ket{a} \ket{b}.
\end{align}
We apply a sequence of operations as follows:
\begin{enumerate}
\item Apply the phase oracle \( U_f \) to the first register: \( \ket{x} \mapsto f(x) \ket{x} \).
    \item Apply \( \mathrm{CADD}_{x \to a} \): \( \ket{x}\ket{a} \mapsto \ket{x}\ket{x+a} \).
    \item Apply \( U_f^\dagger \) to the second register: \( \ket{x+a} \mapsto \overline{f(x+a)} \ket{x+a} \).
    \item Apply inverse \( \mathrm{CADD}_{x \to a} \): undo addition.
    \item Apply \( \mathrm{CADD}_{x \to b} \): \( \ket{x}\ket{b} \mapsto \ket{x}\ket{x+b} \).
    \item Apply \( U_f^\dagger \) to the second register: \( \ket{x+b} \mapsto \overline{f(x+b)} \ket{x+b} \).
    \item Apply \( \mathrm{CADD}_{x \to a+b} \) to a fresh third register (or re-use): result is \( \ket{x+a+b} \mapsto f(x+a+b)\ket{x+a+b} \).
    \item Undo all controlled additions to return to \( \ket{x}\ket{a}\ket{b} \).
\end{enumerate}

After these operations, the overall global phase is given by
\begin{align}
\Delta_{a,b} f(x) = f(x) \cdot \overline{f(x+a)} \cdot \overline{f(x+b)} \cdot f(x+a+b),
\end{align}
and the resulting quantum state is
\begin{align}
\ket{\psi} = \frac{1}{N^{3/2}} \sum_{x,a,b \in G} \Delta_{a,b}f(x) \ket{x} \ket{a} \ket{b}.
\end{align}

We then apply the quantum Fourier transform \( \mathrm{QFT}_G^{\otimes 3} \), which acts as
\begin{align}
\ket{x} \ket{a} \ket{b} \mapsto \frac{1}{N^{3/2}} \sum_{\gamma_1, \gamma_2, \gamma_3 \in \widehat{G}} \gamma_1(x) \gamma_2(a) \gamma_3(b) \ket{\gamma_1} \ket{\gamma_2} \ket{\gamma_3},
\end{align}
where each \( \gamma_i \in \widehat{G} \) is a character, i.e., a homomorphism from \( G \) to \( \mathbb{C}^\times \).

The final quantum state becomes (recall Equation \ref{eq:U2})
\begin{align}
\ket{\psi'} =   \sum_{\gamma_1, \gamma_2, \gamma_3 \in \widehat{G}} \left( \mathbb{E}_{x,a,b} \left[ \Delta_{a,b}f(x) \cdot \gamma_1(x) \gamma_2(a) \gamma_3(b) \right] \right) \ket{\gamma_1} \ket{\gamma_2} \ket{\gamma_3}.
\end{align}

In particular, the amplitude at the trivial character \( \gamma_i = 0 \) for all \( i \) is
\begin{align}
\braket{0^{\otimes 3} | \psi'} = \mathbb{E}_{x,a,b} \Delta_{a,b}f(x) = \|f\|_{U^2}^4,
\end{align}
so the measurement probability of observing \( \ket{0}^{\otimes 3} \) is
\begin{align}
\Pr[0^{\otimes 3}] = \left( \mathbb{E}_{x,a,b} \Delta_{a,b}f(x) \right)^2 = \|f\|_{U^2}^8.
\end{align}

Therefore, by repeating the measurement sufficiently many times and computing the empirical frequency of the all-zero outcome, one can estimate \( \|f\|_{U^2} \) up to additive error \( \varepsilon \) using \( O(1/\varepsilon^2) \) repetitions. Let us count the number of cost. The phase oracle \( U_f \) is queried $4$ times. Each addition operation is group addition, assumed efficient (e.g., polynomial in \( \log \max(N_i) \)). The quantum Fourier transform over \( G^{\otimes 3} \) \cite{childs2010quantum} has gate complexity \( O(\log^2 \max(N_i) ) \) when \( G \cong \mathbb{Z}_{N_1} \times \cdots \times \mathbb{Z}_{N_k} \), using efficient QFT circuits for each component.
\end{proof}\subsection{Generalization to Gowers \( U^d \) Norm}

We now generalize the above construction to the Gowers \( U^d \) norm.

\begin{theorem}[Quantum Algorithm for Gowers \( U^d \) Norm]\label{thm:qk}
Let \( f: G \to \mathbb{C} \) with $|f|=1$ and \( d \geq 2 \). There exists a quantum algorithm that
\begin{itemize}
    \item makes \( 2^d \) queries to the phase oracle \( U_f \) or its conjugate,
    \item performs \( d+1 \) quantum Fourier transforms over \( G \),
    \item outputs a quantum state such that the probability of measuring \( \ket{0}^{\otimes (d+1)} \) is exactly \( \|f\|_{U_d}^{2^{d+1}} \).
\end{itemize}
The total gate complexity is
\begin{align}
\mathcal{O}\left( 2^d C_f + (d+1) \cdot \log^2 |G| \right).
\end{align}
\end{theorem}

\begin{proof}
The proof follows the same pattern as for the \( U^2 \) case. The Gowers \( U^d \) norm involves a sum over \( 2^d \) terms of the form
\begin{align}
\prod_{\omega \in \{0,1\}^d} f\left(x + \omega \cdot \vec{a} \right)^{(-1)^{|\omega|}}.
\end{align}
Each term requires querying \( f \) or \( \overline{f} \), determined by the parity of \( |\omega| \). After preparing the uniform superposition over \( (x, a_1, \ldots, a_d) \), we compute all \( 2^d \) relevant points and apply the corresponding phase oracles or conjugates. We then apply the quantum Fourier transform over \( G^{\otimes (d+1)} \) and measure. The amplitude at \( \ket{0}^{\otimes (d+1)} \) yields \( \|f\|_{U_d}^{2^{d+1}} \), as required.
\end{proof}


\section{Applications}\label{sec:applications}

For the application, we mainly focus on $\mathbb{F}_p^n$, although our quantum algorithm for estimating Gowers norms (Theorem~\ref{thm:qk}) works for any finite abelian group. The reason for focusing on $\mathbb{F}_p^n$ is that the inverse theorem for the Gowers norm can be stated more cleanly in the finite field case, without involving nilsequences.

We now demonstrate a simple yet important application: distinguishing structured functions from random ones using Gowers norms. The Gowers \( U^{d+1} \) norm is known to be $1$ for phase polynomials of degree at most \( d \), and exponentially small for random functions with high probability. This makes it a natural tool for quantum property testing. 
\subsection{Distinguishing Phase Polynomials from Random Functions}

Let \( G = \mathbb{F}_p^n \), and let \( \omega = e^{2\pi i/p} \). Suppose we are given oracle access to a function \( f : G \to \mathbb{C} \) satisfying \( |f(x)| = 1 \) for all \( x \in G \). We consider the following promise problem:

\begin{enumerate}
    \item[\textbf{(Yes)}] \( f(x) = \omega^{P(x)} \) for some unknown polynomial \( P : \mathbb{F}_p^n \to \mathbb{F}_p \) of degree at most \( d \).
    \item[\textbf{(No)}] \( f \) is a Haar-random function from \( \mathbb{F}_p^n \) to the unit circle in \( \mathbb{C} \), i.e., each \( f(x) \sim \text{Uniform}(S^1) \), independently at random.
\end{enumerate}

It is known that for Haar-random functions, the Gowers \( U^{d+1} \) norm is small with high probability:
\begin{equation}
    \mathbb{E}_f \left[ \|f\|_{U^{d+1}}^{2^{d+1}} \right] \leq \frac{1}{p^n},
\end{equation}
(see, e.g., \cite{tao2012higher, tao2010inverse}) indicating that the all-zero amplitude in our quantum algorithm will be close to zero.
\begin{theorem}[Quantum Detection of Degree-\( d \) Phase Polynomials via Gowers Norm]
\label{thm:quantum-detection}
Fix a constant \( d > 1 \). Let \( f : \mathbb{F}_p^n \to \mathbb{C} \) be a function with \( |f(x)| = 1 \). There exists a quantum algorithm which, using \( O(2^d) \) queries to the phase oracle \( U_f : \ket{x} \mapsto f(x)\ket{x} \), decides with high confidence whether \( f \) is a degree-\( d \) phase polynomial or a completely random function.

Specifically:
\begin{align}
    \text{The algorithm accepts if and only if } \|f\|_{U^{d+1}} = 1.
\end{align}
Moreover, the algorithm estimates \( \|f\|_{U^{d+1}}^2 \) via amplitude estimation:
\begin{align}
    \Pr[\text{output} = \ket{0}^{\otimes (d+1)n}] = \|f\|_{U^{d+1}}^{2^{d+1}}.
\end{align}
\end{theorem}

\begin{proof}
This is a direct application of the quantum algorithm described in Theorem~\ref{thm:qk}. If \( f = \omega^{P(x)} \) for some degree-\( d \) polynomial \( P \), then \( \|f\|_{U^{d+1}} = 1 \), so the output probability is \( 1 \). On the other hand, for a Haar-random function \( f \), the expected value \( \mathbb{E}_f[\|f\|_{U^{d+1}}^{2^{d+1}}] \leq p^{-\Omega(n)} \), i.e., exponentially small in \( n \). Therefore, a constant number of repetitions suffices to distinguish the two cases with high probability.
\end{proof}

Classically, the well-known Blum–Luby–Rubinfeld (BLR) test provides a method to check whether a Boolean function \( f : \mathbb{F}_2^n \to \{-1,1\} \) is linear. This test can be interpreted as an estimation of the Gowers \( U^2 \) norm of \( f \). Recent work such as \cite{jothishwaran2020quantum} has extended this idea to the quantum setting, demonstrating how quantum algorithms can be used to estimate the Gowers norm to test linearity over \( \mathbb{F}_2^n \).

In this work, we extend this quantum testing framework in two important directions:
\begin{enumerate}
    \item \textbf{From \( \mathbb{F}_2 \) to general \( \mathbb{F}_p \):} We generalize the quantum linearity testing algorithm from the binary field \( \mathbb{F}_2 \) to any prime field \( \mathbb{F}_p \), allowing for broader applicability in non-Boolean settings.
    \item \textbf{From linear to higher-degree phase polynomials:} We extend the framework from detecting linear functions to detecting degree-\( d \) phase polynomials over \( \mathbb{F}_p^n \) for any \( d \geq 1 \), using the Gowers \( U^{d+1} \) norm as a natural generalization.
\end{enumerate}

We say a function \( f : \mathbb{F}_p^n \to \mathbb{C} \), with \( |f(x)| = 1 \), is \( \varepsilon \)-far from all degree-\( d \) phase polynomials if
\begin{equation}
\forall P \in \mathbb{F}_p[x_1, \dots, x_n], \deg(P) \leq d,\quad \left| \langle f, \omega^{P} \rangle \right| < \varepsilon.
\end{equation}
In this case, the Gowers norm \( \|f\|_{U^{d+1}} \) must be small, and our quantum algorithm acts as a property tester that distinguishes phase structure from random-like behavior.

\subsection{Quantum Approximate Detection of Degree-\( d \) Phase Polynomials}

We now extend our exact detection result to an approximate version, under the promise that the function \( f \) is either a degree-\( d \) phase polynomial or \( \varepsilon \)-far from all such functions. The analysis relies on a well-known inverse theorem for the Gowers norm, summarized in Terence Tao's exposition \cite{tao2012higher}. For completeness, we recall the following:

\begin{theorem}[1\% Inverse Theorem for \( U_{d+1} \) Norm (Theorem 1.5.3 in \cite{tao2012higher})]
Let \( \mathbb{F} \) be a field of characteristic greater than \( d \), and let \( V = \mathbb{F}_p^n \). Suppose \( f : V \to \mathbb{C} \) satisfies \( \|f\|_\infty \leq 1 \) and \( \|f\|_{U_{d+1}} \geq \varepsilon \). Then there exists a degree-\( d \) polynomial \( P : V \to \mathbb{F}_p \) such that
\begin{align}
|\langle f, \omega^{P} \rangle| \geq \delta,
\end{align}
for some constant \( \delta = \delta(p, d, \varepsilon) > 0 \).
\end{theorem}

By the contrapositive of this theorem, if \( f \) is \( \varepsilon \)-far from all degree-\( d \) phase polynomials—that is, if for all such \( P \), we have \( |\langle f, \omega^P \rangle| < \varepsilon \)—then the Gowers norm must be small:
\begin{align}
\|f\|_{U_{d+1}} \leq \delta = \delta(\varepsilon, d, p).
\end{align}
Note that the function \( \delta(\varepsilon, d, p) \) depends intricately on higher-order Fourier analysis and is generally not known in closed form.

We now present our main result, which generalizes the quantum linearity test of \cite{jothishwaran2020quantum} (valid for \( p=2 \)) to arbitrary prime fields and higher-degree polynomials.

\begin{theorem}\label{thm:qn}[Quantum Approximate Detection of Degree-\( d \) Phase Polynomials]
Let \( f : \mathbb{F}_p^n \to \mathbb{C} \) be a function with \( |f(x)| = 1 \) for all \( x \), given as a phase oracle. Assume \( p > d \). Then there exists a quantum algorithm that distinguishes:
\begin{enumerate}
    \item[(Yes)] \( f(x) = \omega^{P(x)} \) for some degree-\( d \) polynomial \( P \), versus
    \item[(No)] \( f \) is \( \epsilon \)-far from all degree-\( d \) phase polynomials,
\end{enumerate}
with error at most \( \eta \), using
\begin{align}
O\left(\frac{1}{\Delta^2} \log \frac{1}{\eta} \right), \quad \text{where } \Delta := 1 - \delta^{2^{d+1}}
\end{align}
measurements on a state constructed using \( 2^d \) queries to the oracle \( U_f \) and \( d+1 \) QFT gates over \( \mathbb{F}_p^n \).
\end{theorem}

\begin{proof}
Following the construction from the exact case, we prepare the quantum state
\begin{align}
\frac{1}{N^{(d+1)/2}} \sum_{x, h_1, \ldots, h_d \in \mathbb{F}_p^n} \Delta_{h_1, \ldots, h_d} f(x) \ket{x}\ket{h_1}\ldots\ket{h_d},
\end{align}
and apply \( d+1 \) parallel QFTs. The amplitude at the all-zero state \( \ket{0}^{\otimes (d+1)n} \) equals
\begin{align}
\mathbb{E}_{x, h_1, \ldots, h_d} \Delta_{h_1, \ldots, h_d} f(x) = \|f\|_{U_{d+1}}^{2^{d+1}}.
\end{align}
In the YES case, this value is exactly 1. In the NO case, by the inverse theorem, it is at most \( \delta^{2^{d+1}} \).

To distinguish between these two scenarios, we measure the quantum state multiple times and estimate the empirical frequency \( \hat{p} \) of observing the zero string. Let \( p \) be the true probability of measuring \( \ket{0}^{\otimes (d+1)n} \). We aim to distinguish \( p = 1 \) versus \( p \leq \delta^{2^{d+2}} \). By a standard Chernoff bound:
\begin{align}
\Pr\left[ |\hat{p} - p| \geq \frac{\Delta}{2} \right] \leq 2 \exp\left( - \frac{m \Delta^2}{2} \right),
\end{align}
where \( \Delta := 1 - \delta^{2^{d+2}} \) and \( m \) is the number of independent measurements. To ensure total error probability at most \( \eta \), it suffices to choose
\begin{align}
m = O\left( \frac{1}{\Delta^2} \log \frac{1}{\eta} \right).
\end{align}
\end{proof}

\begin{corollary}[Total Resources]
Let \( C_f \) be the cost of a single query to the oracle \( U_f \), and assume the cost of a QFT over \( \mathbb{F}_p^n \) is \( O(\log^2 |G|) \). Then the overall complexity of the algorithm is:
\begin{align}
\text{Oracle complexity: } O\left(2^d \cdot \frac{1}{\Delta^2} \log \frac{1}{\eta} \right),
\end{align}
\begin{align}
\text{Gate complexity: } O\left( (2^d C_f + d \log^2 |G|) \cdot \frac{1}{\Delta^2} \log \frac{1}{\eta} \right).
\end{align}
\end{corollary}

\paragraph{Remarks.} Due to the existence of nonclassical polynomials for fields with small characteristic, the inverse theorem does not generally hold for \( p \leq d \). Therefore, our algorithm assumes \( p > d \). However, for \( d = 1 \) and \( d = 2 \) (i.e., testing linear and quadratic phase polynomials), the inverse theorem is known to hold for all prime fields \( \mathbb{F}_p \), and our algorithm applies without restriction.

The effectiveness of our quantum test hinges on the availability of explicit bounds on the Gowers norm in the absence of low-degree structure. In the next section, we begin with the linear case \( d=1 \), where such bounds can be obtained via standard Fourier analysis.
\subsection{Application to Linearity Testing over \( \mathbb{F}_p^n \)}

We now demonstrate a concrete application of our framework to the case \( d = 1 \), i.e., testing whether a function is linear over \( \mathbb{F}_p^n \). This setting corresponds to detecting linear phase polynomials, and serves as a natural generalization of the quantum linearity test over \( \mathbb{F}_2^n \) proposed \cite{jothishwaran2020quantum}. Here, we extend their result to arbitrary prime fields \( \mathbb{F}_p \), for all \( p \geq 2 \), showing that efficient quantum detection is still possible using the \( U^2 \) norm.

In the linear case, an explicit bound on the Gowers norm can be derived using standard Fourier analysis. We begin with the following lemma:

\begin{lemma}\label{lem:lf}
Let \( f : \mathbb{F}_p^n \to \mathbb{C} \) be a function with \( |f(x)| \leq 1 \) for all \( x \), and suppose that
\begin{align}
\left| \left\langle f, \chi \right\rangle \right| = \left| \frac{1}{|G|} \sum_{g \in G} f(g)\chi(g) \right| \leq \varepsilon
\quad \text{for all additive characters } \chi : \mathbb{F}_p^n \to \mathbb{C}.
\end{align}
Then the Gowers \( U^2 \) norm of \( f \) satisfies
\begin{align}
\|f\|_{U^2} \leq \varepsilon^{1/2}.
\end{align}
\end{lemma}

\begin{proof}
The Gowers \( U^2 \) norm is related to Fourier coefficients via:
\begin{align}
\|f\|_{U^2}^4 = \sum_{\xi \in \mathbb{F}_p^n} |\widehat{f}(\xi)|^4,
\end{align}
where \( \widehat{f}(\xi) = \mathbb{E}_{x \in \mathbb{F}_p^n} f(x) \omega^{-\langle \xi, x \rangle} \) is the Fourier coefficient at frequency \( \xi \). By assumption, \( |\widehat{f}(\xi)| \leq \varepsilon \) for all \( \xi \). Hence,
\begin{align}
\|f\|_{U^2}^4 = \sum_{\xi} |\widehat{f}(\xi)|^4 
\leq \varepsilon^2 \sum_{\xi} |\widehat{f}(\xi)|^2 
= \varepsilon^2 \cdot \|f\|_2^2 \leq \varepsilon^2,
\end{align}
since \( |f(x)| \leq 1 \) implies \( \|f\|_2^2 \leq 1 \). Taking the fourth root completes the proof:
\begin{align}
\|f\|_{U^2} \leq \varepsilon^{1/2}.
\end{align}
\end{proof}

By combining Lemma~\ref{lem:lf} with our general detection theorem (Theorem~\ref{thm:qn}), we obtain an efficient quantum algorithm for testing linearity over \( \mathbb{F}_p^n \):

\begin{theorem}\label{thm:linear}
[Quantum Approximate Detection of Linear Phase Polynomials]  
Let \( f : \mathbb{F}_p^n \to \mathbb{C} \) be a function with \( |f(x)| = 1 \) for all \( x \), given as a phase oracle. Then there exists a quantum algorithm that distinguishes:
\begin{enumerate}
    \item[(Yes)] \( f(x) = \omega^{P(x)} \) for some linear polynomial \( P \), versus
    \item[(No)] \( f \) is \( \varepsilon \)-far from all linear phase polynomials,
\end{enumerate}
with error at most \( \eta \), using
\begin{align}
O\left(\frac{1}{\Delta^2} \log \frac{1}{\eta} \right), \quad \text{where } \Delta := 1 - \varepsilon^4,
\end{align}
measurements on a quantum state constructed with 2 queries to the oracle \( U_f \) and 3 QFTs over \( \mathbb{F}_p^n \).
\end{theorem}

We close this section by mention a enhancement of Theorem \ref{thm:qn}, as we know that Gower's $U^2$ norm as two way, if $f$ correlates with linear character then Gower's norm is large and vice versa. If correlated is larger that $\epsilon$, then $\|f\|_{U^2}\geq \epsilon$ by definition. So the following is true:
\begin{corollary}
Suppose \( f : \mathbb{F}_p^n \to \mathbb{C} \) is a function with \( |f(x)| = 1 \) for all \( x \), given as a phase oracle. Then there exists a quantum algorithm that distinguishes: (if $\epsilon_1>\epsilon_2^{1/2}$ )
\begin{enumerate}
    \item[(Yes)] There exists a additive character $\chi$ such that $|\braket{f,\chi}|>\epsilon_1$
    \item[(No)] \( f \) is \( \epsilon_2 \)-far from all additive characters $\chi$
\end{enumerate}
with error at most \( \eta \), using
\begin{align}
O\left(\frac{1}{\Delta^2} \log \frac{1}{\eta} \right), \quad \text{where } \Delta := \epsilon_1^8 - \epsilon_2^{4}
\end{align}
measurements on a state constructed with \( 2 \) queries to the oracle \( U_f \) and \( 3 \) QFT gates over \( \mathbb{F}_p^n \).
\end{corollary}
\begin{proof}
Its very simple to show that, in the yes case, the probability to measure $0$ state is at least $\epsilon_1^8,$ in the no case, that the  the probability to measure $0$ state is at most $\epsilon_2^{4}.$ The rest is just the Chernoff bounds.
\end{proof}

\noindent
\textit{Remark.} The proof follows by setting \( d = 1 \) and \( \delta = \varepsilon^{1/2} \), and noting that this inverse bound holds for all \( p \), including \( p = 2 \). Thus, this result generalizes the quantum linearity test of \cite{jothishwaran2020quantum} from the Boolean domain \( \mathbb{F}_2^n \) to arbitrary 
prime fields \( \mathbb{F}_p^n \).

We close this section by mention a enhancement of Theorem \ref{thm:qn}, as we know that Gower's $U^2$ norm as two way, if $f$ correlates with linear character then Gower's norm is large and vice versa. If correlated is larger that $\epsilon$, then $\|f\|_{U^2}\geq \epsilon$ by definition. So the following is true:
\begin{corollary}
Suppose \( f : \mathbb{F}_p^n \to \mathbb{C} \) is a function with \( |f(x)| = 1 \) for all \( x \), given as a phase oracle. Then there exists a quantum algorithm that distinguishes: (if $\epsilon_1>\epsilon_2^{1/2}$ )
\begin{enumerate}
    \item[(Yes)] There exists a additive character $\chi$ such that $|\braket{f,\chi}|>\epsilon_1$
    \item[(No)] \( f \) is \( \epsilon_2 \)-far from all additive characters $\chi$
\end{enumerate}
with error at most \( \eta \), using
\begin{align}
O\left(\frac{1}{\Delta^2} \log \frac{1}{\eta} \right), \quad \text{where } \Delta := \epsilon_1^8 - \epsilon_2^{4}
\end{align}
measurements on a state constructed with \( 2 \) queries to the oracle \( U_f \) and \( 3 \) QFT gates over \( \mathbb{F}_p^n \).
\end{corollary}
\begin{proof}
Its very simple to show that, in the yes case, the probability to measure $0$ state is at least $\epsilon_1^8,$ in the no case, that the  the probability to measure $0$ state is at most $\epsilon_2^{4}.$ The rest is just the Chernoff bounds.
\end{proof}

After geneerlaize the linear testing, next section, we describe some results by apply Theorem \ref{thm:qn} to the higher order cases.
\subsection{Higher-Degree Testing over \( \mathbb{F}_p^n \)}

Having established the case \( d = 1 \), where we generalized the quantum linearity test from the Boolean domain \( \mathbb{F}_2^n \) to arbitrary prime fields \( \mathbb{F}_p^n \), we now turn to the more challenging task of testing for higher-degree structure. Specifically, we apply the general detection framework of Theorem~\ref{thm:qn} to the case where \( f \) is promised to either be a degree-\( d \) phase polynomial or \( \varepsilon \)-far from all such polynomials.

While the \( d=1 \) case admits an explicit and quantitative inverse theorem leading to efficient quantum algorithms, the situation becomes significantly more subtle when \( d \geq 2 \). In such cases, the effectiveness of quantum detection depends on the existence of inverse theorems \ref{thm:inverse} for higher-order Gowers norms. These theorems guarantee that if the Gowers norm \( \|f\|_{U^{d+1}} \) is large, then \( f \) must correlate with some degree-\( d \) phase polynomial. However, the quantitative dependence of this correlation on the norm—namely, the function \( \delta(\varepsilon) \) is known only in special cases and is typically non-constructive or highly inefficient.

In what follows, we summarize known inverse theorems and their quantitative forms (where available), and use them to instantiate our general quantum detection framework for several important cases, including degree-2 (quadratic) and degree-3 (cubic) testing. This highlights the flexibility of our approach: while classical algorithms often rely on explicit structural recovery, our quantum algorithm operates purely via norm amplification and detection, enabling structure testing even in regimes where constructive inverse theorems are unavailable.

These limitations in current inverse theorems highlight the need to focus on special cases where partial or quantitative results are known. In the next subsection, we survey the state of the art for such quantitative inverse bounds—particularly for \( U^3 \), \( U^4 \), and higher norms—over finite fields. These results directly inform the achievable guarantees of our quantum detection framework in the corresponding regimes.

\subsubsection*{Partial results of Quantitative inverse theorem}
For the \( U^2 \) norm and all primes \( p \), we have a linear test which gives a precise and quantitative statement: if a function is uncorrelated with all additive characters, then its \( U^2 \) norm is small. 
For the higher norm, such quantiative bounds are hard to find and has many rich structures. For $p>k,$ which is called high charactersic case, we don't need to worry about the non classical polynomial, otherwse if $p \leq k$, we have the low-characteristic case. Bergelson, Tao, Ziegler prove the inverse theorem for High-characteristic case \cite{bergelson2010inverse}, and also for low characteristic case \cite{green2012inverse}. Those bounds are not effective and even Tower type. For quantative version. In high charactersic case 2017, Gower, Milićević shows the $U^4$ norm \cite{gowers2017quantitative} but double exponential bounds and further remove one extra exponential, A recent breakthrough by Milićević \cite{milicevic2024quasipolynomial} provides a quasipolynomial inverse theorem over \( \mathbb{F}_p^n \), which establishes the existence of a significant correlation with a cubic phase polynomial whenever the \( U^4 \) norm is large:


\begin{theorem}[Inverse theorem for \( U^4 \) norm over \( \mathbb{F}_p^n \), {\cite[Theorem 2]{milicevic2024quasipolynomial}}]
Let \( f : \mathbb{F}_p^n \to \mathbb{C} \) with \( |f(x)| \leq 1 \), and suppose that
\begin{align}
\|f\|_{U^4} \geq c.
\end{align}
Then there exists a (possibly non-classical) cubic polynomial \( q : \mathbb{F}_p^n \to \mathbb{T} \) such that
\begin{align}
\left| \mathbb{E}_{x \in \mathbb{F}_p^n} f(x) e(q(x)) \right| \geq \exp\left(-\log^{O(1)}(1/c)\right).
\end{align}
\end{theorem}
We may also restate the above theorem in contrapositive form, which is more suitable for the design of structure-testing algorithms:
\begin{corollary}[Contrapositive form]\label{cor:4}
Let \( f : \mathbb{F}_p^n \to \mathbb{C} \) with \( |f(x)| \leq 1 \), and suppose that for all cubic phase polynomials \( q : \mathbb{F}_p^n \to \mathbb{T} \), we have
\begin{align}
\left| \mathbb{E}_{x \in \mathbb{F}_p^n} f(x) e(q(x)) \right| \leq \delta.
\end{align}
Then
\begin{align}
\|f\|_{U^4} \leq \exp\left( -(\log(1/\delta))^{\Omega(1)} \right).
\end{align}
\end{corollary}

For $p=2, $ Quantitative inverse theorem for Gowers uniformity norms for $U^5, U^6$ can be also obtained. Similariy, Corollary 6 in \cite{milicevic2022quantitative} 
\begin{theorem}[Inverse theorem for \( U^k \) norm over \( \mathbb{F}_2^n \) for \( k = 5, 6 \)]
Let \( k \in \{5, 6\} \), and let \( f : \mathbb{F}_2^n \to \mathbb{C} \) be a bounded function with \( |f(x)| \leq 1 \), and suppose that
\begin{align}
\|f\|_{U^k} \geq c.
\end{align}
Then there exists a generalized polynomial \( q : \mathbb{F}_2^n \to \mathbb{T} \) of degree at most \( k - 1 \) such that
\begin{align}
\left| \mathbb{E}_{x \in \mathbb{F}_2^n} f(x) \cdot \exp(2\pi i q(x)) \right| \geq \left(\exp^{(O(1)} (O(1/c))\right)^{-1}.
\end{align}
\end{theorem}
Similarly:
\begin{corollary}[Contrapositive form]\label{cor:56}
Let \( f : \mathbb{F}_2^n \to \mathbb{C} \), with \( |f(x)| \leq 1 \), and suppose that for all generalized polynomials \( q : \mathbb{F}_2^n \to \mathbb{T} \) of degree at most \( k - 1 \), we have
\begin{align}
\left| \mathbb{E}_{x \in \mathbb{F}_2^n} f(x) \cdot \exp(2\pi i q(x)) \right| \leq \delta.
\end{align}
Then
\begin{align}
\|f\|_{U^k} \leq \exp\left( -\mathrm{polylog}(1/\delta) \right).
\end{align}
\end{corollary}

Combining Corollaries~\ref{cor:4} and~\ref{cor:56} with the quantum detection framework in Theorem~\ref{thm:quantum-detection}, we obtain the following result, which establishes a quasipolynomial-time quantum algorithm for approximate phase polynomial detection:
\begin{theorem}[Quasipolynomial-Time Quantum Detection of Degree-\( d \) Phase Polynomials]\label{thm:qn}
Let \( d = 3 \) and \( p \) be any prime, or let \( d = 4, 5 \) with \( p = 2 \). Suppose \( f : \mathbb{F}_p^n \to \mathbb{C} \) is a function with \( |f(x)| = 1 \) for all \( x \), accessible via a quantum phase oracle. Then there exists a quantum algorithm that, with error at most \( \eta \), distinguishes between the following two cases:

\begin{enumerate}
    \item[(Yes)] \( f(x) = \omega^{P(x)} \) for some degree-\( d \) polynomial \( P : \mathbb{F}_p^n \to \mathbb{F}_p \);
    \item[(No)] \( f \) is \( \epsilon \)-far from every degree-\( d \) (possibly non-classical) phase polynomial.
\end{enumerate}

The algorithm requires
\begin{align}
O\left(\frac{1}{\Delta^2} \log \frac{1}{\eta} \right) \quad \text{measurements, where} \quad \Delta := 1 - \exp\left( -\mathrm{polylog}(1/\epsilon) \right)^{2^{d+2}},
\end{align}
and each measurement is performed on a quantum state prepared using \( 2^d \) queries to the oracle \( U_f \), along with \( d+1 \) quantum Fourier transforms over \( \mathbb{F}_p^n \).
\end{theorem}

\paragraph{Remark}
To illustrate the asymptotic behavior of our detection algorithm, consider the case when the function \( f \) is \( \epsilon \)-far from any degree-\( d \) phase polynomial, with 
\begin{align}
\epsilon = 1 - \frac{1}{\mathrm{poly}(n)}.
\end{align}
In the linear case \( (d=1) \), the inverse theorem guarantees the detection algorithm distinguishes the YES and NO cases with gap
\begin{align}
\Delta = 1 - \|f\|_{U^2}^2 = 1 - \epsilon^4 = \Theta\left( \frac{1}{\mathrm{poly}(n)} \right).
\end{align}
This leads to an overall query and measurement complexity of
\begin{align}
O\left( \frac{1}{\Delta^2} \log \frac{1}{\eta} \right) = \mathrm{poly}(n),
\end{align}
thus yielding a polynomial-time quantum algorithm.

However, in the higher-order case \( d \geq 3 \), the best known inverse theorems only guarantee
\begin{align}
\|f\|_{U^{d+1}} \leq \delta(\epsilon) \leq \exp\left( -\log^{O(1)}(1/\epsilon) \right),
\end{align}
so that
\begin{align}
\Delta = 1 - \delta(\epsilon)^{2^{d+2}} = 1 - \exp\left( -\log^{O(1)}(1/\epsilon) \right)^{2^{d+2}} = \exp\left( -\log^{O(1)} n \right).
\end{align}
This implies a total complexity of
\begin{align}
O\left( \frac{1}{\Delta^2} \log \frac{1}{\eta} \right) = \exp\left( \log^{O(1)} n \right),
\end{align}
which is quasipolynomial in \( n \). Therefore, even though the quantum algorithm remains conceptually simple and efficient in structure, the current bounds on inverse theorems render the algorithm quasipolynomial in the worst case. This justifies the terminology used in Theorem~\ref{thm:qn}.

We emphasize that this quasipolynomial complexity arises solely from the current limitations in the known inverse theorems for higher-order Gowers norms. Any future improvement in the quantitative bounds of these inverse theorems. For example, establishing that \( \delta(\epsilon) \geq \mathrm{poly}(\epsilon) \) would immediately improve \( \Delta \) and reduce the complexity of our quantum detection algorithm to polynomial time. Thus, the algorithm provides a robust framework that can directly benefit from advances in higher-order Fourier analysis.

We also note that parts of the quasipolynomial bounds over \( \mathbb{Z}/N\mathbb{Z} \) have been extended to the \( U^4 \) norm~\cite{leng2023efficient}, and more recently to general \( U^k \) norms by James Leng, Ashwin Sah, and Mehtaab Sawhney~\cite{leng2024quasipolynomial}. However, for the \( U^3 \) norm, obtaining quantitative bounds remains a major open challenge.

We now move to another interesting applications. The analysis of 3-term arithmetic progressions via the \( U^2 \) norm not only illustrates the power of quantum algorithms informed by Fourier structure, but also lays the groundwork for extending such techniques to more complex patterns. In particular, longer arithmetic progressions and other additive configurations require higher-order Gowers norms, such as \( U^3 \), \( U^4 \), and beyond. In the next section, we investigate the potential and limitations of such generalizations, both from a classical and quantum perspective.

\subsection{Counting for three terms Arthimetic Progressions in $\mathbb{F}_p^n$}\label{subsec:c}
Arithmetic progressions are among the most fundamental objects studied in additive combinatorics and analytic number theory, and they have been a central source of motivation for the development of higher-order Fourier analysis. Historically, Roth's celebrated theorem \cite{Roth53} shows that any subset of the positive integers with positive upper density must contain nontrivial 3-term arithmetic progressions. The same statement holds over finite fields \( \mathbb{F}_p^n \) for \( p \geq 3 \), where it can be proved using standard Fourier analytic techniques.

Later, Szemerédi \cite{Szemeredi75} generalized this result to arbitrary-length arithmetic progressions, introducing the powerful method of density increments. It was eventually recognized that linear Fourier analysis is insufficient to fully control configurations such as 3- and 4-term progressions in more complex settings. This realization led to Gowers' foundational work \cite{Gowers01}, which introduced the \( U^k \) norms as a means to quantitatively capture the uniformity of functions and bound the number of arithmetic progressions.

While further generalizations over cyclic groups require the use of tools such as nilsequences—which are beyond the scope of this work—the behavior of Gowers norms in finite field settings like \( \mathbb{F}_p^n \) remains tractable and rich. In this section, we explore how Gowers norms can be used to analyze arithmetic progressions and how such structure can inform the design of quantum algorithms for property testing problems related to additive combinatorics.

We start the discussion recall the simplisest identity.
\begin{lemma}[Fourier Identity for 3-APs]\label{lem:3ap}
Let \( G = \mathbb{F}_p^n \) be a finite abelian group, and let \( f, g, h : G \to \mathbb{C} \) be bounded functions. Then the following identity holds:
\begin{align}
\mathbb{E}_{x, y \in G} f(x) g(x+y) h(x+2y) = \sum_{\gamma \in \hat{G}} \hat{f}(\gamma)\hat{g}(-2\gamma)\hat{h}(\gamma),
    \end{align}
where \( \hat{f}(\gamma) := \mathbb{E}_{x \in G} f(x) \overline{\chi_\gamma(x)} \) denotes the Fourier coefficient of \( f \) at character \( \gamma \in \hat{G} \), and \( \chi_\gamma \) is the canonical additive character of \( G \).
\end{lemma}

\begin{proof}
We begin by expanding each function using its Fourier expansion:
\begin{align}
f(x) = \sum_{\alpha} \hat{f}(\alpha) \chi_\alpha(x), \quad
g(x+y) = \sum_{\beta} \hat{g}(\beta) \chi_\beta(x+y), \quad
h(x+2y) = \sum_{\delta} \hat{h}(\delta) \chi_\delta(x+2y).
\end{align}
Then, the left-hand side becomes:
\begin{align}
\mathbb{E}_{x,y} f(x) g(x+y) h(x+2y)
&= \mathbb{E}_{x,y} \left( \sum_{\alpha} \hat{f}(\alpha) \chi_\alpha(x) \right)
\left( \sum_{\beta} \hat{g}(\beta) \chi_\beta(x+y) \right)
\left( \sum_{\delta} \hat{h}(\delta) \chi_\delta(x+2y) \right) \nonumber \\
&= \sum_{\alpha, \beta, \delta} \hat{f}(\alpha) \hat{g}(\beta) \hat{h}(\delta)
\mathbb{E}_{x,y} \chi_{\alpha+\beta+\delta}(x) \chi_{\beta+2\delta}(y),
\end{align}
since \( \chi_\beta(x+y) = \chi_\beta(x)\chi_\beta(y) \), and similarly for \( \chi_\delta(x+2y) \).

Now use orthogonality of characters:
\begin{align}
\mathbb{E}_{x \in G} \chi_\lambda(x) =
\begin{cases}
1 & \text{if } \lambda = 0, \\
0 & \text{otherwise}.
\end{cases}
\end{align}
Therefore, the expectation is non-zero if and only if:
\begin{align}
\alpha + \beta + \delta = 0 \quad \text{and} \quad \beta + 2\delta = 0.
\end{align}
Solving these equations gives:
$\delta = \gamma, \quad \beta = -2\gamma, \quad \alpha = \gamma.$ Thus, we obtain:
\begin{align}
\sum_{\gamma \in \hat{G}} \hat{f}(\gamma) \hat{g}(-2\gamma) \hat{h}(\gamma),
\end{align}
as desired.
\end{proof}

\begin{theorem}[Equivalence Between Gowers \( U^2 \) Norm and 3-AP Bias]
\label{thm:u2-3ap-correct}
Let \( f : \mathbb{F}_p^n \to \mathbb{C} \) with \( \|f\|_\infty \leq 1 \), and define
\begin{align}
T(f) := \mathbb{E}_{x,d} f(x) f(x+d) f(x+2d).
\end{align}
Then:
\begin{enumerate}
    \item[(i)] If \( \|f\|_{U^2} \geq \varepsilon \), then
\begin{align}
    |T(f)| \geq \varepsilon^5.
\end{align}
    \item[(ii)] If \( |T(f)| = \alpha \), then
\begin{align}
    \|f\|_{U^2} \geq \alpha^{1/2}.
\end{align}
\end{enumerate}
\end{theorem}

\begin{proof}
Recall from Lemma~\ref{lem:3ap}:
\begin{align}
T(f) = \sum_{\gamma} \hat{f}(\gamma)^2 \hat{f}(-2\gamma).
\end{align}

(i) Suppose \( \|f\|_{U^2}^4 = \sum_\gamma |\hat{f}(\gamma)|^4 \geq \varepsilon^4 \).  
Then, by the inequality
\begin{align}
\sum_\gamma |\hat{f}(\gamma)|^4 \leq \left( \max_\gamma |\hat{f}(\gamma)|^2 \right) \cdot \sum_\gamma |\hat{f}(\gamma)|^2,
\end{align}
and since \( \sum_\gamma |\hat{f}(\gamma)|^2 \leq 1 \) by Parseval, we conclude that:
\begin{align}
\max_\gamma |\hat{f}(\gamma)| \geq \varepsilon^2.
\end{align}
Pick \( \gamma_0 \) such that \( |\hat{f}(\gamma_0)| \geq \varepsilon^2 \), then again using Parseval, there exists \( \gamma_1 \) such that \( |\hat{f}(-2\gamma_0)| \geq \varepsilon \), say. Then:
\begin{align}
|T(f)| \geq |\hat{f}(\gamma_0)^2 \hat{f}(-2\gamma_0)| \geq \varepsilon^4 \cdot \varepsilon = \varepsilon^5.
\end{align}
(ii) Suppose \( |T(f)| \geq \alpha \). Then by the triangle inequality:
\begin{align}
|T(f)| = \left| \sum_\gamma \hat{f}(\gamma)^2 \hat{f}(-2\gamma) \right|
\leq \sum_\gamma |\hat{f}(\gamma)|^2 \cdot |\hat{f}(-2\gamma)|.
\end{align}
Applying Cauchy–Schwarz:
\begin{align}
\sum_\gamma |\hat{f}(\gamma)|^2 \cdot |\hat{f}(-2\gamma)|
\leq \left( \sum_\gamma |\hat{f}(\gamma)|^4 \right)^{1/2}
\cdot \left( \sum_\gamma |\hat{f}(-2\gamma)|^2 \right)^{1/2}
\leq \|f\|_{U^2}^2.
\end{align}
Therefore:
\begin{align}
\alpha \leq \|f\|_{U^2}^2 \Rightarrow  \|f\|_{U^2} \geq \alpha^{1/2} .
\end{align}
\end{proof}
The precise nature of this equivalence is unique to the linear (i.e., \( U^2 \)) regime. In higher-order settings such as \( U^3 \), only one direction remains true: large Gowers norm still implies correlation with structure (e.g., quadratic phases), but the converse becomes significantly more delicate.

We now use Theorem~\ref{thm:u2-3ap-correct} to construct a quantum algorithm for estimating the number of 3-term arithmetic progressions in a Boolean function \( f : \mathbb{F}_p^n \to \{0,1\} \). This approach is fundamentally different from both classical sampling and Grover's quantum counting algorithm. Instead of directly querying membership in a set, we estimate the Gowers \( U^2 \)-norm of \( f \), which captures the additive structure of the function. This offers a new perspective: by leveraging Fourier-analytic techniques, one can design norm-based quantum algorithms that count linear patterns in a dataset.

This method is not as query-efficient as Grover’s optimal quantum counting algorithm in the worst case, but it illustrates an important principle: quantum algorithms that exploit algebraic or spectral structure can sometimes offer more interpretable or structurally adaptive solutions. In particular, it provides a new framework for analyzing the number of arithmetic configurations using quantum amplitude estimation, a tool central to quantum statistical estimation.

\paragraph{Query Complexity and Comparison with Grover.}
We wish to estimate the quantity
\begin{align}
T(f) := \mathbb{E}_{x,d} f(x) f(x+d) f(x+2d)
\end{align}
to relative error \( \epsilon \), i.e., output \( \widetilde{T(f)} \in (1 \pm \epsilon) T(f) \). Since Theorem~\ref{thm:u2-3ap-correct} relates \( T(f) \) and \( \|f\|_{U^2} \), we consider estimating the Gowers norm \( \|f\|_{U^2}^8 \) using amplitude estimation \footnote{
To apply amplitude estimation in a quantum setting, the oracle must be unitary. While the original function \( f \) may be an indicator function with values in \( \{0,1\} \), we can instead define a phase oracle \( U_f(x) = (-1)^{f(x)} \), which is unitary and preserves the relevant structure. This allows us to apply quantum techniques to Boolean functions by embedding them into the unit circle. The remaining discussion assumes such a phase oracle for simplicity, noting that the relationship between \( \widehat{U_f} \) and \( f(x) \) is straightforward.
}

To ensure multiplicative error in \( T(f) \), it suffices to estimate \( \|f\|_{U^2}^8 \) with additive error
\begin{align}
\delta = \epsilon \cdot \|f\|_{U^2}^5,
\end{align}
since \( T(f) \leq \|f\|_{U^2}^2 \) and \( T(f) \geq \|f\|_{U^2}^5 \) imply that
\begin{align}
\frac{1}{\|f\|_{U^2}^{10}} \leq \frac{1}{T(f)^5}.
\end{align}
Amplitude estimation achieves additive error \( \delta \) with query complexity \( O(1/\delta^2) \), so our overall query complexity becomes
\begin{align}
O\left( \frac{1}{\epsilon^2 \cdot \|f\|_{U^2}^{10}} \right) \leq O\left( \frac{1}{\epsilon^2 \cdot T(f)^5} \right).
\end{align}

In contrast, Grover's quantum counting algorithm estimates \( T(f) \), seen as a mean over a size-\( N^2 \) domain, with multiplicative error \( \epsilon \) using only
\begin{align}
O\left( \frac{1}{\epsilon \sqrt{T(f)}} \right)
\end{align}
queries.

\paragraph{Discussion.}
Although the Gowers-norm-based quantum algorithm has worse scaling in \( T(f) \) compared to Grover's method, it is structurally motivated and may offer advantages in contexts where functions exhibit additive or spectral structure. Moreover, this norm-based approach aligns with the broader philosophy of structure-versus-randomness in additive combinatorics, and could serve as a foundation for hybrid algorithms or new applications in property testing and learning theory.

\paragraph{Review of Grover Counting.}
For completeness, we state Grover's counting complexity result:

\begin{theorem}[Grover Counting \cite{aaronson2020quantum, jozsa1999searching}]
Let \( f : [N] \to \{0,1\} \) be a Boolean function with \( M \) solutions. Fix any desired \( \delta > 0 \). Then there exists a quantum algorithm which outputs \( \widetilde{M} \in (1 \pm \epsilon) M \) with success probability at least \( 1 - \delta \), using
\begin{align}
O\left( \sqrt{\frac{N}{M}} \cdot \frac{1}{\epsilon} \cdot \log \frac{1}{\delta} \right)
\end{align}
queries to \( f \).
\end{theorem}

Applying this to the domain \( \mathbb{F}_p^n \times \mathbb{F}_p^n \) (for \( x,d \)), Grover can estimate the number of 3-term arithmetic progressions \( M = |\{(x,d) : f(x)f(x+d)f(x+2d)=1\}| \) with query complexity
\begin{align}
O\left( \sqrt{\frac{N^2}{M}} \cdot \frac{1}{\epsilon} \cdot \log \frac{1}{\delta} \right).
\end{align}

This represents the optimal scaling for unstructured search, as shown by the BBBV bound \cite{bennett1997strengths}. Our Gowers-based method, while suboptimal in general, provides a new structural quantum tool that could inspire further developments.

In this section, we demonstrated how the Gowers \( U^2 \) norm provides an effective quantum approach for estimating the number of 3-term arithmetic progressions in a Boolean function. Though not query-optimal in the general case, this method reveals how algebraic structure can guide quantum algorithm design, offering alternative paradigms to unstructured quantum search. These ideas set the stage for future explorations of higher-order uniformity and its algorithmic implications. Next we show that our quantum algorithms are robust under some noisy models.

\section{Comment on NISQ Devices}\label{sec:nisq}

Due to technological limitations, fault-tolerant quantum computation remains beyond our current capabilities. As a result, the concept of Noisy Intermediate-Scale Quantum (NISQ) computing was introduced by Preskill~\cite{preskill2018quantum} to describe near-term quantum devices that operate without full error correction. 

There are several competing physical implementations of NISQ devices, including trapped ions~\cite{bruzewicz2019trapped}, superconducting qubits~\cite{wallraff2004strong, krantz2019quantum, larsen2015semiconductor}, optical systems~\cite{o2007optical}, and neutral atoms~\cite{henriet2020quantum}. Due to their heterogeneity, no universally accepted formal model of NISQ computation exists.

Recently, efforts to formalize such models have emerged. Chen et al.~\cite{chen2022complexity} proposed a complexity class called $\text{NISQ}$ to capture noisy quantum computations, and Jacobs et al.~\cite{jacobs2024space}, along with Chia et al.~\cite{chia2024oracle}, demonstrated that variants of Bernstein–Vazirani and Forrelation problems can be solved within such models, even under substantial noise. These algorithms share a common technique: quantum Fourier sampling over $\mathbb{F}_2^n$, which closely relates to our Gowers-norm-based methods.

\paragraph{Half-BQP Model.}
We first recall the $\frac{1}{2}\text{BQP}$ model introduced by Aaronson~\cite{aaronson2017computational}:

\begin{definition}
In the $\frac{1}{2} \text{BQP}$ model, the initial state is
\begin{align}
\frac{1}{\sqrt{2^n}} \sum_{x \in \mathbb{F}_2^n} \ket{x} \ket{x},
\end{align}
a maximally entangled state across $2n$ qubits. A quantum circuit is applied, followed by classical $\text{BPP}$ post-processing of measurement outcomes.
\end{definition}

Jacobs et al.~\cite{jacobs2024space} showed that $\frac{1}{2}\text{BQP}$ can simulate Instantaneous Quantum Polynomial-time (IQP) circuits, and solve Simon’s problem, Bernstein–Vazirani, and Forrelation of which provide oracle separations from the Polynomial Hierarchy (PH)~\cite{raz2022oracle}. This model is closely related to the DQC1 (“one clean qubit”) model~\cite{knill1998power}, but differs in that the entanglement structure is symmetric and no single qubit is fully clean.

\paragraph{Noisy Initial State Model.}
Chia et al.~\cite{chia2024oracle} introduced a variant of Chen’s $\text{NISQ}$ class, denoted $\text{BQP}_{\lambda}$, which formalizes a noisy circuit model by incorporating local depolarization noise during the computation. We briefly summarize a restricted version:

\begin{definition}[$\text{BQP}_{\lambda}^{0}$]
Let $\lambda \in (0,1)$. The $\text{BQP}_{\lambda}^{0}$ model assumes that the initial state is a classical mixture of bit-strings, given by:
\begin{align}
\rho = \sum_{e \in \{0,1\}^n} \lambda^{|e|}(1-\lambda)^{n - |e|} \ket{e} \bra{e},
\end{align}
where $|e|$ denotes the Hamming weight of $e$, and each bit flips independently with probability $\lambda$. Computation proceeds via noisy two-qubit gates and noisy measurements, but no intermediate measurements are allowed.
\end{definition}

This model forbids fault-tolerant error correction. When $\lambda = 0$, we recover standard $\text{BQP}$. As shown in~\cite{chia2024oracle}, even for $\lambda = O(1)$, problems like Bernstein–Vazirani and Deutsch–Jozsa remain solvable. Remarkably, when $\lambda = O(\log n / n)$, the model can solve Forrelation with oracle separation from $\text{PH}$.

\paragraph{Comparing $\text{BQP}_{\lambda}^{0}$ and $\frac{1}{2}\text{BQP}$.}
The main distinction is that in $\text{BQP}_{\lambda}^{0}$, the initial state is completely unknown (a noisy classical distribution), whereas in $\frac{1}{2}\text{BQP}$, the initial entangled state is fully known and accessible. However, when $\lambda = 1/2$, the initial state in $\text{BQP}_{\lambda}^{0}$ becomes the uniform distribution over $\{0,1\}^n$, which coincides with the reduced state of each register in $\frac{1}{2}\text{BQP}$. Thus, we have
\begin{align}
\text{BQP}_{1/2}^{0} \subseteq \frac{1}{2}\text{BQP},
\end{align}
though whether this containment is strict remains open. Note also that $\text{BQP}_{\lambda}^{0} \cong \text{BQP}_{1 - \lambda}^{0}$ due to symmetry in the noise model.

We now demonstrate that the Gowers norms (specifically \( \|f\|_{U^2} \)) are invariant under the one particular noisy quantum models discussed above. This robustness arises from the symmetry of finite difference operations.
\begin{theorem}
There exist quantum algorithms within $\frac{1}{2}\text{BQP}$ model that estimate the $U^2$ norm using exactly $4$ queries to the phase oracle and $3$ quantum Fourier transforms (QFTs) over a finite Abelian group $G$. The probability of observing a specific Fourier basis outcome determined by fixed shifts \( \omega_1, \omega_2, \omega_3 \in G \) is exactly \( \|f\|_{U^2}^8 \). The total gate cost is \( 4C_f + 3\log^2|G| \), where \( C_f \) is the cost of implementing the oracle for \( f \).
\end{theorem}

\begin{proof}
The construction is inspired by~\cite{jacobs2024space} and tracks additive noise via fixed shift vectors. We describe the circuit for estimating \( \|f\|_{U^2} \).

Let \( N = |G| \). We initialize the state using noisy inputs and create superposition.
\begin{align}
\ket{\omega_1}\ket{\omega_2}\ket{\omega_3} \to \frac{1}{N^{3/2}}\sum_{x,a,b}\chi_{\omega_1}(x)\chi_{\omega_2}(a)\chi_{\omega_3}(b)\ket{x}\ket{a}\ket{b} 
\end{align}
without knowing $\omega$'s in advance,
where \( \omega_1, \omega_2, \omega_3 \in G \) are fixed but arbitrary shifts introduced by noise in state preparation. Importantly, while these values may be known, we do not rely on them in any computational step.

We perform the exact following operations:
\begin{enumerate}
    \item Apply phase oracle \( U_f \) to the first register.
    \item Apply controlled addition \( \mathrm{CADD}_{x \to a} \), then apply conjugate oracle again on first register.
    \item Undo the controlled addition.
    \item Apply \( \mathrm{CADD}_{x \to b} \), followed by another conjugate oracle.
    \item Apply \( \mathrm{CADD}_{x \to a} \), then apply the final phase oracle.
    \item Undo all additions to return to the original (noisy) basis state.
\end{enumerate}

After these steps, the resulting state is:
\begin{align}
\ket{\psi} = \frac{1}{N^{3/2}}\sum_{x,a,b \in G} \widetilde{\Delta_{a,b}f(x)} \ket{x}\ket{a}\ket{b},
\end{align}
where
\begin{align}
\widetilde{\Delta_{a,b}f(x)} = \chi_{\omega_1}(x)\chi_{\omega_2}(a)\chi_{\omega_3}(b)f(x)\overline{f(x+a)}\overline{f(x+b)}f(x+a+b).
\end{align}
This is a shifted version of the usual finite difference \( \Delta_{a,b} f(x) \) resulting from the shifted input states.

We now apply the quantum Fourier transform over \( G^{\otimes 3} \), resulting in:
\begin{align}
\ket{\psi_1} = \sum_{\chi_{g}, \chi_{h}, \chi_{k} \in \hat{G}} \left( \frac{1}{N^3}\sum_{x,a,b \in G} \chi_g(x)\chi_h(a)\chi_k(b) \cdot \widetilde{\Delta_{a,b}f(x)} \right) \ket{\chi_g}\ket{\chi_h}\ket{\chi_k}.
\end{align}
The proability of measuring $\ket{\chi_g}\ket{\chi_h}\ket{\chi_k}$ is the following:
\begin{align}
 &\left( \frac{1}{N^3}\sum_{x,a,b \in G} \chi_g(x)\chi_h(a)\chi_k(b) \cdot \widetilde{\Delta_{a,b}f(x)} \right)^2 \nonumber \\
&= \left( \frac{1}{N^3}\sum_{x,a,b \in G} \chi_g(x)\chi_h(a)\chi_k(b) \chi_{\omega_1}(x)\chi_{\omega_2}(a)\chi_{\omega_3}(b)f(x)\overline{f(x+a)}\overline{f(x+b)}f(x+a+b) \right)^2 
\end{align}

Due to the additive shifts in input, the Fourier-domain amplitude is not concentrated at \( \ket{0}^{\otimes 3} \), but instead at \( \ket{\chi_{\omega_1^{-1}}, \chi_{\omega_2^{-1}}, \chi_{\omega_3^{-1}}} \), where each \( \chi_{\omega_i^{-1}} \) is the character corresponding to the inverse shift \( \omega_i^{-1} \). Specifically, the amplitude at that point is:
\begin{align}
\braket{\chi_{\omega_1^{-1}}, \chi_{\omega_2^{-1}}, \chi_{\omega_3^{-1}} | \psi_1} = \mathbb{E}_{x,a,b} \Delta_{a,b}f(x),
\end{align}
since for each \( j=1,2,3 \), the identity \( \chi_{\omega_j^{-1}}(x)\chi_{\omega_j}(x) = 1 \) holds by character property. The squared magnitude of this amplitude equals \( \|f\|_{U^2}^8 \), as in the noiseless setting.

Because the Gowers norm is invariant under additive shifts, this peak remains valid for norm estimation. Although the outcome is shifted in the Fourier basis, we can still use amplitude estimation to estimate the corresponding probability and thus accurately compute the norm.
Moreover, since the shifts \( \omega_i \) are fixed (though arbitrary) and we know its at the final state and the algorithm does not rely on them at any step, this construction is valid under \( \frac{1}{2}\text{BQP} \).

\end{proof}

\paragraph{Remark.}
The analysis above assumes that the additive shifts \( \omega_i \in G \) are fixed but arbitrary—either due to static device imperfections or bounded noise in state preparation. This suffices for correctness in the \( \frac{1}{2}\mathrm{BQP} \) and noise-resilient \( \mathrm{BQP}_\lambda^0 \) models under deterministic shifts.

However, in the fully randomized noise setting where the \( \omega_i \) are drawn independently from an unknown distribution \( \mathcal{D} \), the peak in the Fourier spectrum will fluctuate across different runs of the algorithm. As a result, the amplitude estimation algorithm may no longer sample a consistent output, and the probability mass spreads across multiple basis elements. In this case, estimating \( \|f\|_{U^2}^8 \) requires averaging over shifted Fourier outputs, potentially degrading the convergence rate or requiring error mitigation techniques.

We leave a full analysis of this random-noise model, and possible derandomization or symmetrization strategies, for future work.

\section{Conclusion}\label{sec:conclusion}

In this work, we introduced quantum algorithms for estimating the Gowers \( U^d \)-norm over arbitrary finite abelian groups, extending classical higher-order Fourier analytic tools into the quantum regime. Our constructions utilize phase oracles and quantum Fourier transforms, and we provide explicit bounds on the oracle complexity and gate depth for the resulting circuits.

As a key application, we used our quantum \( U^2 \)-norm estimation algorithm to test linearity over \( \mathbb{F}_p^n \), thereby generalizing the classical Blum–Luby–Rubinfeld (BLR) test to a quantum setting over larger fields. This recovers and extends known quantum linearity tests, and provides a unifying perspective for property testing through norm estimation.

Building on recent inverse theorems for Gowers norms~\cite{milicevic2022quantitative, milicevic2024quasipolynomial}, we further designed quasipolynomial-time quantum algorithms for detecting degree-\( d \) phase polynomials for \( d = 3, 4, 5 \) in appropriate regimes. Our framework unifies structure-versus-randomness techniques from additive combinatorics with quantum algorithms, and is, to our knowledge, the first quantum realization of higher-order uniformity testing beyond the Boolean domain.

A crucial feature of our approach is its robustness under certain types of noise. Because Gowers norms are invariant under additive shifts, our algorithms remain valid even in the presence of fixed but unknown shift errors. This ensures that the output probability distribution is unaffected by such misalignments, allowing implementation in restricted quantum models such as \( \frac{1}{2}\text{BQP} \), where correctness is guaranteed on a \( 1/2 \) fraction of noise realizations. While this does not yet establish correctness under general randomized noise distributions as required by the \( \text{BQP}_\lambda^0 \) model, it offers a partial robustness guarantee that may already be meaningful in NISQ implementations, especially where calibration errors or drift behave in a quasi-static manner.

Several promising directions emerge. While our analysis focused on vector spaces \( \mathbb{F}_p^n \), extending to more general finite abelian groups—particularly those requiring nilsequence-based inverse theorems—would be a natural next step~\cite{elsholtz2017number, ferenczi2018ergodic}. Moreover, adapting our techniques to non-abelian or continuous groups, such as compact Lie groups, could bridge our work with other quantum algorithms for algebraic and number-theoretic tasks, including those computing Littlewood–Richardson coefficients~\cite{larocca2025quantum, bravyi2025classical}, Kronecker coefficients~\cite{bravyi2024quantum}, and Gauss sums~\cite{van2002efficient}.

Our results suggest that Gowers-norm-based quantum algorithms may serve as a flexible primitive in quantum property testing, learning theory, and the analysis of pseudorandomness—especially in settings where robustness to noise is paramount.






\bibliographystyle{unsrt}
\bibliography{sample}

\end{document}